\documentclass[a4paper,english]{lipics-v2019}

\usepackage{graphicx}
\usepackage{graphics}
\usepackage{amssymb}
\usepackage{amsmath}
\usepackage{amsthm}
\usepackage{amsfonts}
\usepackage{latexsym}
\usepackage{psfrag}
\usepackage{enumerate}
\usepackage{here}
\usepackage{esvect}
\usepackage[update,prepend]{epstopdf}
\usepackage{anyfontsize}
\usepackage{epsfig}
\usepackage{epsf}
\usepackage{fancyvrb}
\usepackage{url}
\usepackage{hyperref}

\usepackage{microtype}

\newtheorem*{observation*}{Observation}

\newcommand{\etal}{{et~al.}}
\newcommand{\ie}{{i.e.}}
\newcommand{\eg}{{e.g.}}

\newcommand{\NN}{\mathbb{N}} 

\def\Med{{\tt Med}}
\def\rank{{\rm rank}}
\def\p{{\tt prefix}}
\def\pred{{\tt pred}}
\def\succ{{\tt succ}}

\newcommand{\later}[1]{}
\newcommand{\old}[1]{}

\nolinenumbers 
\title{Multiparty Selection}

\titlerunning{Multiparty Selection}

\author{Ke Chen}
{Department of Computer Science,
University of Wisconsin--Milwaukee, USA}
{kechen@uwm.edu}
{0000-0001-5470-6621}
{}
\author{Adrian Dumitrescu}
{Department of Computer Science,
University of Wisconsin--Milwaukee, USA}
{dumitres@uwm.edu}
{0000-0002-1118-0321}%
{}

\authorrunning{K.~Chen, A.~Dumitrescu}

\Copyright{Ke Chen, Adrian Dumitrescu}



\begin{document}

\maketitle

\begin{abstract}
Given a sequence $A$ of $n$ numbers and an integer (target) parameter $1\leq i\leq n$, 
the (\emph{exact}) selection problem asks to find the $i$-th smallest element in $A$. 
An element is said to be $(i,j)$-\emph{mediocre} if it is neither among the top $i$
nor among the bottom $j$ elements of $S$. The \emph{approximate} selection problem asks to
find a $(i,j)$-mediocre element for some given $i,j$; as such, this variant allows
the algorithm to return any element in a prescribed range.
In the first part, we revisit the selection problem in the two-party model introduced
by Andrew Yao (1979) and then extend our study of exact selection to the multiparty model.
In the second part, we deduce some communication complexity benefits that arise
in approximate selection. In particular, we present a deterministic protocol for finding
an approximate median among $k$ players. 
\end{abstract}



\section{Introduction} \label{sec:intro}

Given a sequence $A$ of $n$ numbers and an integer (selection) parameter $1\leq i\leq n$, 
the selection problem asks to find the $i$-th smallest element in $A$.
If the $n$ elements are distinct, the $i$-th smallest is larger than
$i-1$ elements of $A$ and smaller than the other $n-i$ elements of $A$.
By symmetry, the problems of determining the $i$-th smallest and the $i$-th largest
are equivalent. 
Together with sorting, the selection problem is one of the most fundamental problems in
computer science.  Whereas sorting trivially solves the selection problem in $O(n \log{n})$ time,
Blum~\etal~\cite{BFP+73} gave an $O(n)$-time algorithm for this problem.

The selection problem, and computing the median in particular are in close relation
with the problem of finding the quantiles of a set. 
The $h$-th \emph{quantiles} of an $n$-element
set are the $h-1$ order statistics that divide the sorted set in $h$
equal-sized groups (to within $1$); see, \eg, \cite[p.~223]{CLRS09}.
The $h$-th quantiles of a set can be computed by a recursive algorithm running
in $O(n \log{h})$ time.

The selection problem, determining the median in particular, have been also considered
from the perspective of communication complexity in the \emph{two-party} model introduced
by Andrew Yao~\cite{Yao79}.
Suppose that Alice and Bob hold subsets $A$ and $B$ of $[n]=\{1,2,\ldots,n\}$, respectively,
and wish to determine the median of the multiset $A \cup B$ while keeping their communication
close to a minimum. Several classic protocols going back to 1980s achieve this task by exchanging
$O(\log^2{n})$ bits~\cite{KN97,Ro82}. The communication complexity for this task
has been subsequently reduced to $O(\log{n})$ bits~\cite{KN97,RY19}.

\subparagraph{Mediocre elements.}
Following Frances Yao~\cite{Yao74}, an element is said to be $(i,j)$-\emph{mediocre} if it
is neither among the top (\ie, largest) $i$ nor among the bottom (\ie, smallest) $j$
of a totally ordered set $S$ of $n$ elements.
As remarked by Yao, finding a mediocre element is closely related to finding the median,
in the sense that the common goal is selecting an
element that is not too close to either extreme. In particular,
$(i,j)$-mediocre elements where
$i=\lfloor \frac{n-1}{2} \rfloor$, $j= \lfloor \frac{n}{2} \rfloor$
(and symmetrically exchanged), are medians of~$S$. 
Previous work on \emph{approximate selection} (in this sense) includes~\cite{BCC+00,Du19}.

In Section~\ref{sec:approx-2} we study the communication complexity of finding a
mediocre element in the two-party model introduced by Andrew Yao~\cite{Yao79}.
The communication complexity of finding the median in this model has
been studied in~\cite{Ro82,CT87,Ma93}; see also~\cite{KN97,RY19}. 
In particular we outline a scenario in which computing a mediocre element
near the median can be accomplished with communication complexity~$O(1)$---which is
very attractive.

\subparagraph{Background and related problems.}
Due to its primary importance, the selection problem has been studied extensively;
see for instance~\cite{AKSS89,BCC+00,BJ85,CM89,DHU+01,DZ96,DZ99,FR75,FG79,
  HS69,Ho61,Hy76,J88,KKZZ18,Ki81,Pa96,SPP76,Yap76}.
A comprehensive review of early developments in selection is provided by Knuth~\cite{Kn98}.
The reader is also referred to dedicated book chapters on selection, such as those 
in~\cite{AHU83,Ba88,CLRS09,DPV08,KT06} and the more recent articles~\cite{CD15,Du16},
including experimental work~\cite{Al17}. 

In many applications (\eg, sorting), it is not important to find an exact median,
or any other precise order statistic, for that matter, and an approximate median
suffices~\cite{EW19}. For instance, quick-sort type algorithms aim at finding
a (not necessarily perfect) balanced partition rather quickly; see \eg, \cite{BCC+00,Ho61}.

\subparagraph{Our results.} Our main results are summarized in the three theorems stated below.
We first study the communication complexity of finding the median in the multiparty setting. 
In this model we assume that every message by one of the players is seen by all the players
(\ie, it is a broadcast); as in~\cite[p.~83]{KN97}. 

\begin{theorem} \label{thm:exact-k}
  For $i=1,\ldots,k$, let player $i$ hold a sequence (\ie, a multiset) $A_i$ whose
  support is a subset of $[n]$ and $|A_i|=O\left(\text{poly}(n)\right)$.
  There is a deterministic protocol for finding the median of $\uplus_{i=1}^k A_i$
  (\ie, their multiset sum) with $O(k \log^2{n})$ communication complexity.
\end{theorem}

We then study the communication complexity of finding an approximate median
in the multiparty setting (under slightly stronger assumptions on the input sets). 

\begin{theorem} \label{thm:approx-k}
Let $\alpha=p/q$, where $p<q/2$, $q$ is fixed and $0<c \leq 1$ be a positive constant.
For $i=1,\ldots,k$, let player $i$ hold a set $A_i \subset [n]$ that is disjoint from
any other player's set.
Assume that $t=|\cup_{i=1}^k A_i| \geq cn$. Put $\ell= \lceil \log{\frac{2q}{c}} \rceil$.
Then an $(\alpha t, \alpha t)$-mediocre element of $\cup_{i=1}^k A_i$ can be found with
$O(\ell \cdot k \log{n}) =O(k \log{n})$ communication complexity.
\end{theorem}

In particular, a $(t/3, t/3)$-mediocre element, or a $(0.49 \, t, 0.49 \, t)$-mediocre element,
among $k$ players can be determined with $O(k \log{n})$ communication complexity.

In the final part of our paper, somewhat surprisingly, we show that
(under suitable additional assumptions and a somewhat relaxed requirement)
the communication complexity of finding a mediocre element in the vicinity of the median
is bounded from above by a constant and is therefore independent of $n$.

\begin{theorem} \label{thm:approx-2}
Let $\alpha=p/q$, where $p<q/2$, $q$ is fixed and $0<c \leq 1$ be a positive constant.
Let Alice and Bob hold disjoint sets $A$ and $B$ of elements from $[n]$,
where $s=|A| \leq |B|=m$.
Let $t=s+m$ denote the total number of elements in $A \cup B$, where $t \geq cn$.
Assume that $t$, $c$, and $\alpha$ are known to both of them.
Put $h = \lceil \frac{2q}{q-2p} \rceil$ and $\ell= \lceil \log{\frac{12h}{c}} \rceil$.

Then an $(\alpha t, \alpha t)$-mediocre element can be found (by at least one player)
with $O(\ell \, \log{h}) =O(1)$ communication complexity.
If both players return, each element returned is $(\alpha t, \alpha t)$-mediocre;
the elements found by the players need not be the same. 
\end{theorem}

In particular, a $(t/3, t/3)$-mediocre element, or a $(0.49 \, t, 0.49 \, t)$-mediocre element,
between $2$ players can be determined (by at least one player) with $O(1)$ communication complexity.
A simple example that falls under the scenario in Theorem~\ref{thm:approx-2} is one
where $A$ consists of distinct odd numbers and $B$ consists of distinct even numbers. 

\subparagraph{Preliminaries.} 
A simple but effective procedure reduces the selection problem for finding the $i$-th smallest
element out of $n$ to one for finding the median in a slightly larger sequence.
The target is the $i$-th smallest element in an input sequence $A$ of size $n$. 
Assume first that $i<n/2$;
in this case pad the input $A$ with $n-2i$ elements that are less than or equal to
the minimum in the input sequence; call $A'$ resulting sequence.
Note that $|A'|=n+(n-2i)=2(n-i)$. 
It suffices to observe that the median of $A'$ is the $i$-th smallest element in $A$:
indeed, $n-2i+i=n-i$, as required.
The case $i>n/2$ is symmetric; in this case pad the input $A$ with $2i-n$ elements that are
larger than or equal to the maximum in the input sequence; call $A'$ resulting sequence.
Note that $|A'|=n+(2i-n)=2i$. Observe that the median of $A'$ is the $i$-th smallest element in $A$,
as required. We therefore restrict our attention to the median selection problem.

\subparagraph{Notation.} 
Without affecting the results, the floor and ceiling functions are omitted
in some instances where they are not essential.  
For example, we frequently write the $\alpha n$-th element instead of
the more precise $\lfloor \alpha n \rfloor$-th element. 
Unless specified otherwise, all logarithms are in base~$2$. 


For an $s$-bit number $x$ and a positive integer $\ell$, where $s \geq \ell$,
$\p_\ell(x)$ denotes the \emph{$\ell$-bit binary prefix} of $x$,
\ie, the number formed by the first (\ie, most significant) $\ell$ bits of $x$. 

If $x$ belongs to a sorted list and is not the minimum, $\pred(x)$ denotes its predecessor.
If $x$ belongs to a sorted list and is not the maximum, $\succ(x)$ denotes its successor.

\section{Exact selection} \label{sec:exact}

In this section we prove Theorem~\ref{thm:exact-k}. First, we set up the problem in the
context of two-party communication complexity; we start with some background. 
In this section, each player's input is allowed to contain duplicates.
Following the literature, we refer to these (potential) multisets as sets, and the union
operation should be understood as multiset sum~\cite[Example 1.6, p.~6]{KN97}.
(An equivalent formulation is \emph{merging of sequences}.)

\subsection{Two players} \label{ssec:two}

Alice and Bob hold multisets $A$ and $B$ whose supports are subsets 
of $[n]=\{1,2,\ldots,n]$, respectively.
  It is assumed that  $|A|,|B| =O\left(\text{poly}(n)\right)$.
  (In a standard setup~\cite[Example 1.6, p.~6]{KN97}, $A$ and $B$ are subsets of $[n]$;
  here we extend this setup for potentially larger multisets.)
The median of the multiset $A \cup B$ is denoted by $\xi=\Med(A,B)$;
as usual, the median of $X$ is the $\left\lceil \left(|X|/2 \right) \right\rceil$-th
smallest element of~$X$. 

There is a simple binary-search type protocol due to M.~Karchmer that takes
$O(\log^2{n})$ bits of communication; see~\cite[Example 1.6, p.~6]{KN97}. 
At each round Alice and Bob have an interval $[i,j]$, $i,j \in \NN$,
that contains the median. They halve the interval (repeatedly) by deciding
whether the median is less than, equal to, or larger than $m=(i+j)/2$.
This is done by Alice sending to Bob the number of elements in $A$
that are less than $m$, equal to $m$, and larger than $m$, using $O(\log{n})$ bits.
Bob can now determine whether the median is less than, equal to, or larger than $m$,
and sends this information to Alice using $O(1)$ bits. The protocol has $O(\log{n})$ rounds,
each requiring $O(\log{n})$ bits of communication, so the overall communication complexity
is $O(\log^2{n})$. 
  
An alternative binary-search type protocol that takes $O(\log^2{n})$ bits of communication,
also due to Karchmer~\cite[p.~168]{KN97}, works as follows.
Assume, without loss of generality that $|A|=|B|$ and that the common size is a power of $2$:
this can be achieved by exchanging the sizes of their inputs ($O(\log{n})$ bits) and padding them
with the appropriate number of the minimal element ($1$) and the maximal element ($n$). 
The protocol works in rounds. During the protocol, Alice maintains a set $A' \subset A$ 
of elements that may still be the median (initially $A'=A$) and Bob maintains
a set $B' \subset B$ of elements that may still be the median (initially $B'=B$).
At each round, Alice sends Bob the value $a$, which is the median of $A'$, and
Bob sends Alice the value $b$, which is the median of $B'$.
At this point we have $\min(a,b) \leq \xi \leq \max(a,b)$. 
If $a<b$, then Alice discards the lower half of $A'$ (note that $a$ is part of it)
and Bob discards the upper half of $B'$.
If $b<a$, then Bob discards the lower half of $B'$ (note that $b$ is part of it)
and Alice discards the upper half of $A'$. 
In either case, this operation maintains the median of $A' \cup B'$ as the desired median of
$A \cup B$. It should be noted that the size of $A' \cup B'$ is reduced (exactly) by a factor of $2$. 
If $a=b$, this value is the median, and if $|A'|=|B'|=1$, then the smaller number is the median.
The protocol has $O(\log{n})$ rounds, each requiring $O(\log{n})$ bits of communication,
and so the communication complexity is $O(\log^2{n})$.

The  communication complexity of finding the median can be further reduced. A subtle refinement
of the above protocol, due to Karchmer~\cite[Example 1.7, p.~6 and p.~168]{KN97}, 
and revised by Gasarch~\cite{Book}, 
works with $O(\log{n})$ communication complexity: its key idea is to make
comparisons in a bit-by-bit manner, but this requires careful bookkeeping
of the progress and here we omit the technical details.

\smallskip
We next describe a different (folklore) protocol,
running with $O(\log{n})$ communication complexity,
that we find simpler and subsequently refine for computing a mediocre element. 
The protocol implements a binary-search strategy and works in rounds.
Alice maintains a set $A' \subset A$ of elements that may still
be the median (initially $A'=A$) and Bob maintains a set $B' \subset B$ of elements that may still
be the median (initially $B'=B$).
Alice and Bob compute the medians of their current inputs ($a$ and $b$, respectively).
At this point we have $\min(a,b) \leq \xi \leq \max(a,b)$.
Alice and Bob aim to determine the order relation between $a$ and $b$
in order to halve their input in an appropriate manner. 

The protocol avoids sending these $\log{n}$-bit numbers at each round
by avoiding making a direct comparison between $a$ and $b$.
The players have an interval $[i,j]$, $i,j \in \NN$,
that contains the median (initially, $[i,j]=[1,n]$). 
The medians $a$ and $b$ are compared to the middle element $h= \lfloor (i+j)/2 \rfloor$, 
If $a=b=h$, this element is the median of $A \cup B$ and the protocol terminates.
Otherwise, if $a$ and $b$ are split by $h$, \ie, $a \leq h \leq b$ or 
$b \leq h \leq a$, then (by transitivity of $\leq$), the relation between $a$ and $b$ is
determined, and Alice and Bob halves their input accordingly (as in the earlier $O(\log^2{n})$
protocol). Otherwise, if $a$ and $b$ are on the same side of $h$, \ie,  $a, b \leq h$ or
$h \leq a, b$. For example, in the first case, the elements in the lower half of $A' \cup B'$
are $\leq h$ and the same holds for the median of $A' \cup B'$.
As such, both players shrink their common interval $[i,j]$ by (roughly) half:
the resulting interval is $[i,h]$ or $[h,j]$, respectively. 
Alice and Bob communicate each of the outcomes of the above tests in $O(1)$ bits. 
Each halving operation for $A'$ and $B'$ maintains the property that
$\xi = \Med(A \cup B) = \Med(A' \cup B')$. 

Let $\ell=\lceil \log{n} \rceil$. Note that after $2 \ell-1$ tests, either
Alice and Bob hold singleton sets (\ie, $|A'|=|B'|=1$), or
the common interval $[i,j]$ consists of a single integer $i=j$. 
If $|A'|=|B'|=1$, the smaller number is the median (or either, for equality), 
whereas if $i=j$, this number is the median. 
The number of bits exchanged before the last round of the protocol is $O(\log{n})$
and is $O(\log{n})$ in the last round. The resulting communication complexity is $O(\log{n})$.

\subsection{$k$ players} \label{ssec:k}

In this subsection we show the protocol that proves Theorem~\ref{thm:exact-k}.
It is worth noting that the number of players, $k$ is independent of $n$. 
The protocol maintains the invariant: the median of $\cup_{i=1}^k A_i$ in one round is
the same for the updated sets in the next round.
It is possible that the number of sets drops from $k$ to a lower number;
the protocol remains unchanged until the value $k=2$ is reached,
when the respective players apply the protocol in Subsection~\ref{ssec:two};
recall that padding with extra elements may be needed.
If the value $k=1$ is reached, the remaining player computes the median in his/her
own set and the game ends.

Initially, each player sorts his/her input set locally.
The sorted order is used by each player in the pruning process,
and if such action occurs, the sorted order is locally maintained.
Each set pruning discards elements at one of the two ends of the chain
(either low elements below some threshold, or high elements above some threshold). 

The protocol roughly halves the size of at least one of the current participating sets;
more precisely, for some $X \in \{A_1,\ldots,A_k\}$, we have
$|X'| \leq \left \lfloor |X|/2 \right \rfloor$ by the end of each round. 
Since the size of each set is initially $O\left(\text{poly}(n)\right)$, 
the size of each of the $k$ sets drops to $0$ in at most $O(\log{n})$ iterations
and consequently, the number of rounds is at most $O(k \log{n})$.
(Padding with extra elements when $k=2$ is reached conforms with this bound.)

Each round of the protocol works as follows.
Each player (locally) finds the median of his/her current set: $x_i \in A_i$, $i=1,\ldots,k$. 
The following scheme regarding medians is used: assume that there are $x$ sets of even size
and $y$ sets of odd size in the current round, where $x+y=k$; for the $x$ sets of even size
the first $\lceil x/2 \rceil$ use the lower median and the remaining $\lfloor x/2 \rfloor$
use the upper median (in some fixed, \eg, alphabetical, order). 
The idea of intermixing upper and lower medians is also present in~\cite{CD15}. 
(A scheme that uses only lower medians or only upper medians fails to guarantee 
that the median of the union is maintained after pruning,
for instance if $k=3$ and all three sets have even size; 
the smallest example of this kind is $|A_1|=|A_2|=|A_3|=2$.)

In the first round, each player posts his/her median on the communication board; this involves
$O(k \log{n})$ bits of communication.
In the remaining rounds, two players whose sets got pruned (as further explained below) need
to update their median on the communication board.
Depending on the parities of the sets of these two players before and after the pruning, at most
one more player may need to update his/her median to maintain the balanced scheme adopted earlier
which requires $\lceil x/2 \rceil$  use the lower median and the remaining
$\lfloor x/2 \rfloor$ use the upper median.
Therefore, in each round, the communication complexity is $O(\log{n})$.

All players are now able to determine the sorted order of the $k$ medians. 
For simplicity, assume that after relabeling, this order is
\begin{equation} \label{eq:AB}
  x_1 \leq x_2 \leq \ldots \leq x_k. 
\end{equation}
It is convenient to refer to the players holding the minimum and maximum of these medians 
as Alice and Bob and to their corresponding sets as $A$ and $B$:
$x_A \equiv x_1$ and $x_B \equiv x_k$ (this relabeling is only done for the purpose of analysis).

Let $P$ denote the poset made by the $k$ chains $A_1,\ldots,A_k$, together with the relations in
\eqref{eq:AB}. Write $a=|A|$, $b=|B|$, and $t=\sum_{i=1}^k |A_i|$.
The player holding the smaller set between Alice and Bob 
is in charge of the pruning operation in the current round: the same number of elements
is discarded by Alice and Bob as specified below. Refer to Fig~\ref{fig:1}. 

If $\min(a,b)=a$, Alice discards $\lceil a/2 \rceil$ elements in $A$ (all $x \leq x_A$ when $a$ 
is odd or $x_A$ is the lower median, or all $x < x_A$ when $x_A$ is the upper median),
and Bob discards the highest $\lceil a/2 \rceil$ elements in $B$.
Such operation is \emph{charged} to Alice.
Otherwise, if $\min(a,b)=b$, Bob discards $\lceil b/2 \rceil$ elements in $B$
(all $x \geq x_B$ when $b$ is odd
or $x_B$ is the upper median, or all $x > x_B$ when $x_B$ is the lower median).
Such operation is \emph{charged} to Bob.
It is worth noting that this scheme is feasible: \ie, if the indicated player discards
the specified number of elements, the other player can also discard the same number of elements.
Then the protocol continues with the next round.
Each player keeps track of the players that are still in the game and their set cardinalities,
as these can be deduced from the actions of the algorithm.

\begin{figure}[htbp]
\centering
\includegraphics[scale=0.89]{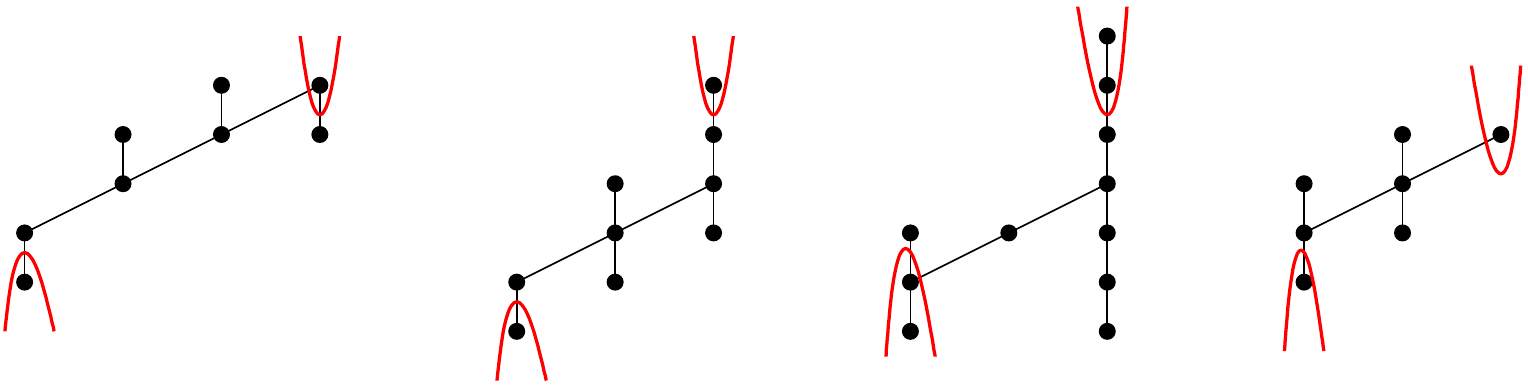}
\caption{Pruning the poset $P$ in the protocol for finding the median;
  Alice is the leftmost player and Bob is the rightmost player.
  (i) $k=4$, $t=8$, $u=6$, $v=5$; operation is charged to Alice.
  (ii) $k=3$, $t=9$, $u=6$, $v=7$; operation is charged to Alice.
  (iii) $t=11$, $u=6$, $v=8$; operation is charged to Alice.   
  (iv) $t=7$, $u=5$, $v=4$; operation is charged to Bob.}
\label{fig:1}
\end{figure}

It remains to show that the same number of elements is discarded from each side of the median
in each round. 
Let $u$ be the number of elements in $P$ that are above the highest discarded element of $A$,
and $v$ be the number of elements in $P$ that are below the lowest discarded element of $B$.
By slightly abusing notation, let $k$ denote the number of players in the current round
of the protocol (which may differ from the initial number). Specifically we prove the following.

\begin{lemma} \label{lem:discard}
  Consider a round of the protocol and assume that $k \geq 3$ and $t=\sum_{i=1}^k |A_i|$.
  The following inequalities for $u$ and $v$ hold:
  $u \geq \lceil \frac{t+1}{2} \rceil$ and $v \geq \lceil \frac{t}{2} \rceil$.
\end{lemma} 
\begin{proof}
For $u$, we start by including $|A_i|/2$ corresponding to the upper half elements in the 
set $A_i$, for $i=1,\ldots, k$; this contributes $t/2$ to the sum. 
In addition we add $1/2$ for each set of odd size, thus $y/2$ over all odd sets.
Then we add $1$ for each set of even size that uses the lower median,
thus $\lceil x/2 \rceil$ over all even sets.
This procedure overcounts by $1$ if the median $x_A$ is the highest discarded element of $A$.
Therefore, we have
\[
u \geq \frac{t}{2}+\frac{y}{2}+\left\lceil \frac{x}{2}\right\rceil -1 
\geq \frac{t}{2}+\frac{y}{2}+ \frac{x}{2} -1 
=\frac{t+x+y-2}{2} =\frac{t+k-2}{2} \geq \frac{t+1}{2}.
\]

Similarly, for $v$, we start by including $|A_i|/2$ corresponding to the lower half elements in the 
set $A_i$, for $i=1,\ldots, k$; this contributes $t/2$ to the sum. 
In addition we add $1/2$ for each set of odd size, thus $y/2$ over all odd sets.
Then we add $1$ for each set of even size that uses the upper median,
thus $\lfloor x/2 \rfloor$ over all even sets.
This procedure overcounts by $1$ if the median $x_B$ is the lowest discarded element of $B$.
Therefore, we have
\[
v \geq \frac{t}{2}+\frac{y}{2}+\left\lfloor \frac{x}{2}\right\rfloor -1
\geq \frac{t}{2}+\frac{y}{2}+\frac{x-1}{2}-1
=\frac{t+x+y-3}{2} =\frac{t+k-3}{2} \geq \frac{t}{2}.
\]

Since both $u$ and $v$ are integers, we have thereby proved that
$u \geq \lceil \frac{t+1}{2} \rceil$ and $v \geq \lceil \frac{t}{2} \rceil$, as required.
\end{proof}

\subparagraph{Proof of Theorem~\ref{thm:exact-k}.}
By Lemma~\ref{lem:discard}, all the elements discarded from $A$ are below the median (of the union),
and all elements discarded from $B$ are above the median.
Thus in each round, the protocol preserves the median and discards the same number of elements
from each side of it. This proves the invariant of the protocol.
Since the protocol takes $O(k \log{n})$ rounds and the communication complexity of each
round is $O(\log{n})$, the overall communication complexity is $O(k \log^2{n})$, as claimed.
\qed

\section{Approximate selection with $k$ players} \label{sec:approx-k}

In this section we consider the problem of finding an $(\alpha t, \alpha t)$-mediocre element
among $k$ players, where $\alpha \in (0,1/2)$ is a fixed constant.
Recall that in the setting of Theorem~\ref{thm:approx-k}, the sets $A_i$, $i=1,\ldots, k$,
are pairwise disjoint. But we do \emph{not} assume that they have the same cardinality.

The protocol works in rounds.
Let $a_1=1$ and $b_1=n$; and note that $[a_1,b_1]$ 
contains the median $m$, \ie, the $\left \lceil t/2 \right \rceil$-th smallest element 
of $\cup_{i=1}^k A_i$.
For round $j=1,2,\ldots$,
the interval $[a_{j+1},b_{j+1}]$ is obtained from
the interval $[a_j, b_j]$ by halving while
maintaining the invariant that $m\in [a_{j+1},b_{j+1}]$.
Equivalently, the invariant can be stated as follows.
For $j=1,2,\ldots$,
\begin{itemize}
\item
the number of elements in $\cup_{i=1}^k A_i$ that are 
$\leq a_j$ is less than $\left\lceil t/2 \right \rceil$,
\item
the number of elements in $\cup_{i=1}^k A_i$ that are 
$\leq b_j$ is at least $\left \lceil t/2 \right \rceil$. 
\end{itemize}

More specifically, in round $j$, let
\[ c_{j} = \left \lfloor \frac{a_{j} + b_{j}}{2} \right \rfloor. \]
Each player communicates the number of elements in his/her set that
are $\leq c_{j}$. Since there are $k$ players, this takes $O(k \log{n})$ bits.
Once this is done, each player can compute independently (by adding the $k$
individual counts) the total number of elements in $\cup_{i=1}^k A_i$ that are 
$\leq c_{i}$.
If the number is less than $\left \lceil t/2 \right \rceil$,
then we set $[a_{j+1},b_{j+1}] := [c_{j},b_j]$,
otherwise, \ie, the number is at least $\left \lceil t/2 \right \rceil$,
then we set $[a_{j+1},b_{j+1}] := [a_j,c_{j}]$.
This setting maintains the invariant.

The protocol repeatedly halves the current interval until
\begin{equation} \label{eq:short}
  b_j - a_j \leq \left( \frac12 - \alpha \right) t.
\end{equation}
When this occurs, since $\cup_{i=1}^k A_i$ consists of distinct elements,
$[a_j,b_j]$ contains a continuous range of no more than $\left( \frac12 - \alpha \right) t$
elements of $\cup_{i=1}^k A_i$, with $m$ being one of them.
If $(0.5-\alpha)t<1$, then the protocol stops when $b_j-a_j=1$ and returns $b_j$ as the median.

Let $z$ be any element of $\cup_{i=1}^k A_i$ contained in $[a_j,b_j]$. 
(The protocol will return one such element, as explained below.)
Observe that 
\begin{align} \label{eq:rank-any}
  \frac{t}{2} - \left( \frac12 - \alpha \right) t
  &\leq \rank_{\cup A_i}(z) \leq
  \frac{t}{2} + \left( \frac12 - \alpha \right) t, \text{ or } \nonumber\\
\alpha t &\leq \rank_{\cup A_i}(z) \leq \left(1 - \alpha \right) t.
\end{align}

The number of halving rounds needed to achieve the interval-length
in~\eqref{eq:short} is at most
\begin{align*}
  \left \lceil \log \frac{n}{\left( \frac12 - \alpha \right) t} \right \rceil 
  &\leq \left \lceil \log \frac{n}{\left( \frac12 - \alpha \right) cn} \right \rceil =
  \left \lceil \log \frac{1}{\left( \frac12 - \alpha \right) c}  \right \rceil 
  = \left \lceil \log \frac{2q}{(q-2p) c} \right \rceil  \\
  &\leq \left \lceil \log \frac{2q}{c}  \right \rceil = \ell = O(1). 
\end{align*}

In each round, the $k$ players communicate their counts, $O(k \log{n})$ bits in total. 
Each player independently computes the total count for the midpoint of the current interval,
and all players take the same decision on how to set the next interval in the halving process
(with no further communication needed).

In the last round (\ie, when inequality~\eqref{eq:short} is satisfied), the players report in turn.
If the player does not hold any element in the interval $[a_j,b_j]$,
he/she outputs a zero bit and the report continues;
otherwise the player outputs such an element (from his/her set) in $O(\log{n})$ bits and
the protocol ends. The output element is a valid choice, as justified by~\eqref{eq:rank-any}. 

The total communication complexity is therefore
$ O(\ell \, k \log{n}) = O(k \log{n})$ bits, as claimed. 
This concludes the proof of Theorem~\ref{thm:approx-k}. 
\qed

\section{Approximate selection with two players under special conditions} \label{sec:approx-2}

Let $t=s+m$ denote the total number of elements in $A \cup B$. 
Here we consider the problem of finding an $(\alpha t, \alpha t)$-mediocre element between two players,
where $\alpha \in (0,1/2)$ is a fixed constant. The protocol described in Subsection~\ref{ssec:two}
immediately yields the following.

\begin{corollary} \label{cor:mp:cc}
  The deterministic communication complexity of finding an $(\alpha t, \alpha t)$-mediocre
  element in $A \cup B \subset [n]$, where $t=|A|+|B|$ and 
  $\alpha \in (0,1/2)$ is a fixed constant, is  $O(\log{n})$.
\end{corollary}

Interestingly enough, this communication complexity can be brought down to a constant
under slightly stronger assumptions:
(i)~$A$ and $B$ have no duplicates or common elements, and
(ii)~$|A \cup B| \geq cn$, for some constant $c>0$;
and a somewhat relaxed requirement:
at least one of the players returns an element to the process that has invoked
his/her service; each element returned is $(\alpha t, \alpha t)$-mediocre.
Note that this is a natural relaxation --- if the set of one player does not contain any suitable element,
it is impossible to communicate the final answer to this player within $O(1)$ complexity.

A natural protocol to consider would be to choose one of the median-finding protocols
and execute a constant number of rounds from it. However, this seemingly promising idea
does not appear to work. It is possible that one of the two sets, say $A$, does not contain
any desired elements, namely  $(\alpha t, \alpha t)$-mediocre for the given $\alpha$
and so at the end of the modified protocol only $B'$ contains desired elements (and not $A'$).
More importantly, the players apparently have no indication of which player is the lucky one.
We therefore resort to a different idea of using quantiles (more precisely, a sampling technique
with a similar effect).

\subparagraph{Proof of Theorem~\ref{thm:approx-2}.}
We may assume, without loss of generality that $n$ and $1/c$ are powers of $2$
(in particular, $4n$ is also a power of $2$). 
For $n < 8q^2/c$ Alice and Bob use the earlier $O(\log{n})$-protocol for finding the median; 
we therefore subsequently assume that $n \geq 8q^2/c$.
In particular, since $q \geq 3$, we have $n \geq 24q/c$.
We further assume, without loss of generality that $|A|=|B|=m$: this can be achieved
by padding the smaller size set with the appropriate numbers of small elements and
large elements as described below. 
In particular, the padding elements need also be distinct.   
(It is \emph{not} assumed that the common size is a power of $2$: since our protocol
does not exactly halve the current set of each player at each round, such an assumption
would be of no use.)

To illustrate the padding process for arbitrary set sizes, 
we may assume without loss of generality that the given input satisfies: $s=|A| \leq |B|=m$. 
We need to pad Alice's input with
$m -\lceil \frac{m+s}{2} \rceil$ small elements and
$\lceil \frac{m+s}{2} \rceil -s$
large elements.
Alice and Bob replace their inputs by $A+n$ and $B+n$, respectively;
as a result, the elements they hold are now in the range $\{n+1,\ldots,2n\}$.
Then Alice pads her input with
$\{1,2,\ldots,m -\lceil \frac{m+s}{2} \rceil\} \subset [n]$ and
$\{2n+1,\ldots,2n+ \lceil \frac{m+s}{2} \rceil -s\} \subset [3n] \setminus [2n]$.
(Note that
$ \left\lceil \frac{m+s}{2}\right\rceil-s = m-\left\lfloor\frac{m+s}{2}\right\rfloor$.)
The resulting sets have the same size $m$ and $A \cup B$ consists of distinct elements
in the range $[3n] \subset [4n]$.
By subtracting $n$, the element(s) returned by the protocol are shifted back
to the original range $[n]$ in the end (without explicitly mentioning it there). 

$A$ and $B$ below denote the (new) padded sets (of size $m$). 
Set $h = \lceil \frac{2q}{q-2p} \rceil$ (recall that $\alpha=p/q$)
and $\ell= \lceil \log{\frac{12h}{c}} \rceil$.
By the assumption $n \geq 24q/c$ we have
\[ cn \geq 24q \geq 12 \left \lceil \frac{2q}{q-2p} \right \rceil =12h. \]
Let $Q_A$ be the set consists of the $i \lfloor m/h \rfloor$-th elements of $A$, for $i=1,2,\ldots,h$
(this set resembles the $h$-th quantiles of $A$).  Similarly,
let $Q_B$ be the set consists of the $i \lfloor m/h \rfloor$-th elements of $B$, for $i=1,2,\ldots,h$
(this set resembles the $h$-th quantiles of $B$).
Note that $|Q_A|=|Q_B|=h$. 
Since $A$ and $B$ consist of pairwise distinct elements, 
between any two elements in $Q_A$ (or $Q_B$), there are at least
\[ \left \lfloor \frac{m}{h}  \right \rfloor \geq \frac{m}{h} -1
\geq \frac{t}{2h} -1 \geq \frac{cn}{2h} -1 \geq \frac{cn}{3h} \geq \frac{4n}{2^\ell} \]
elements. 
Represent each element $x$ in $Q_A$ (and $Q_B$) with $\log (4n) = \log{n} + 2$ bits;
it follows that the elements in $\{\p_\ell(x) : x \in Q_A\}$ are pairwise distinct;
similarly the elements in $\{\p_\ell(y) : y \in Q_B\}$ are pairwise distinct. 

The protocol implements a binary-search strategy aimed at finding the median of 
$Q_A \cup Q_B$. Note that $|Q_A| =|Q_B| \leq h$.  
Alice maintains a set $Q'_A \subset Q_A$ of elements that may still
be the median quantile (initially $Q'_A=Q_A$)
and Bob maintains a set $B' \subset B$ of elements that may still be the median quantile
(initially $Q'_B=Q_B$). The invariant $|Q'_A| =|Q'_B|$ will be maintained. 
At each round, Alice and Bob compute the medians of their current sets
($x_A$ and $x_B$, respectively).
If $\p_\ell(x_A)< \p_\ell(x_B)$ or $\p_\ell(x_A)> \p_\ell(x_B)$ the protocol continues
with Alice and Bob halving their input as in the median-finding protocol.
Specifically, if $\p_\ell(x_A)< \p_\ell(x_B)$ the protocol discards the 
$\left \lfloor |Q'_A|/2 \right \rfloor$ lower elements of $Q'_A$ and the 
$\left \lfloor |Q'_B|/2\right \rfloor$ upper elements of $Q'_B$.
The equality case $\p_\ell(x_A) = \p_\ell(x_B)$ is addressed below.
Observe that the above comparison can be resolved by exchanging $\ell$ bits in each round.

If $\p_\ell(x_A) = \p_\ell(x_B)$, and $|Q'_A|=|Q'_B| \geq 3$, we have
$\p_\ell(\pred(x_A)) <  \p_\ell(x_B)$, and the protocol discards the 
$\left \lfloor (|Q'_A|-1)/2 \right \rfloor$ lower elements of $Q'_A$ and the 
$\left \lfloor (|Q'_B|-1)/2\right \rfloor$ upper elements of $Q'_B$.
Note that this is a slight but important deviation from the standard median-finding protocol;
it is aimed at handling prefix equality by discarding possibly one fewer element by each player.
With this choice, the median of $Q_A \cup Q_B$ remains the median of $Q'_A \cup Q'_B$;
and the invariant $|Q'_A| =|Q'_B|$ is maintained.
Since the sets the players hold are almost halved at each round, the
protocol terminates in $O(\log{h})$ rounds, as specified below.

If $|Q'_A|=|Q'_B| = 2$, and $\p_\ell(x_A) \neq \p_\ell(x_B)$, the protocol continues
with  each player halving his/her own current set accordingly.
If $|Q'_A|=|Q'_B| = 2$, and $\p_\ell(x_A) = \p_\ell(x_B)$, the protocol
terminates with each player output his/her number ($x_A$ and $x_B$, respectively). 
Observe that in this case, the median of $Q_A \cup Q_B$ is $x_A$ or $x_B$
and it will be shown below, see~\eqref{eq:rank}, 
that both elements are $(\alpha t, \alpha t)$-mediocre.

If $|Q'_A|=|Q'_B| =1$ and $\p_\ell(x_A) \neq \p_\ell(x_B)$, the protocol
terminates with the player that holds the smaller of $x_A$ and $x_B$ output that number. 
If $|Q'_A|=|Q'_B| =1$ and $\p_\ell(x_A) = \p_\ell(x_B)$, the protocol
terminates with each player output his/her number ($x_A$ and $x_B$, respectively).
It will be shown below, see~\eqref{eq:rank}, 
that both elements are $(\alpha t, \alpha t)$-mediocre.

Recall that $\ell= \lceil \log{\frac{12h}{c}} \rceil$.
If $x,y \in [3n]$ and $\p_\ell(x) = \p_\ell(y)$ then
\begin{equation} \label{eq:xy}
  |x-y| \leq \frac{3n}{2^\ell} \leq \frac{cn}{4h} \leq \frac{t}{4h}.
\end{equation}

Recall that the median of $Q_A \cup Q_B$ is in $Q'_A \cup Q'_B$
in the last round of the protocol.
Since all elements are distinct, for $x_A$ and $x_B$ above,
if  $\p_\ell(x_A) = \p_\ell(x_B)$, Inequality~\eqref{eq:xy} implies
\begin{equation} \label{eq:rank-xy}
  |\rank_{A \cup B}(x_A) -\rank_{A \cup B}(x_B)| \leq \frac{t}{4h}.
\end{equation}

Assume that the median of $Q_A \cup Q_B$ is $x_A \in Q_A$; then Alice returns $x_A$.
In addition, if  $\p_\ell(x_A) = \p_\ell(x_B)$, Bob also returns $x_B \in Q_B$.
Since $x_A$ is the median of $Q_A \cup Q_B$, it is the $h$-th smallest element of $Q_A \cup Q_B$.
As such (by construction):
(i)~$x_A$ is $\geq$ than at least
\[ h  \left \lfloor \frac{m}{h}  \right \rfloor \geq h \left( \frac{m}{h} -1 \right) =m - h \]
elements of $A \cup B$;
and similarly, (ii)~$x_A$ is $\leq$ than at least $m-h$ elements of $A \cup B$.
Note that the median of $A \cup B$ has rank $m$ and is the same as the median of the
original union of the two sets. See Fig.~\ref{fig:2}.

\begin{figure}[htbp]
\centering
\includegraphics[scale=0.87]{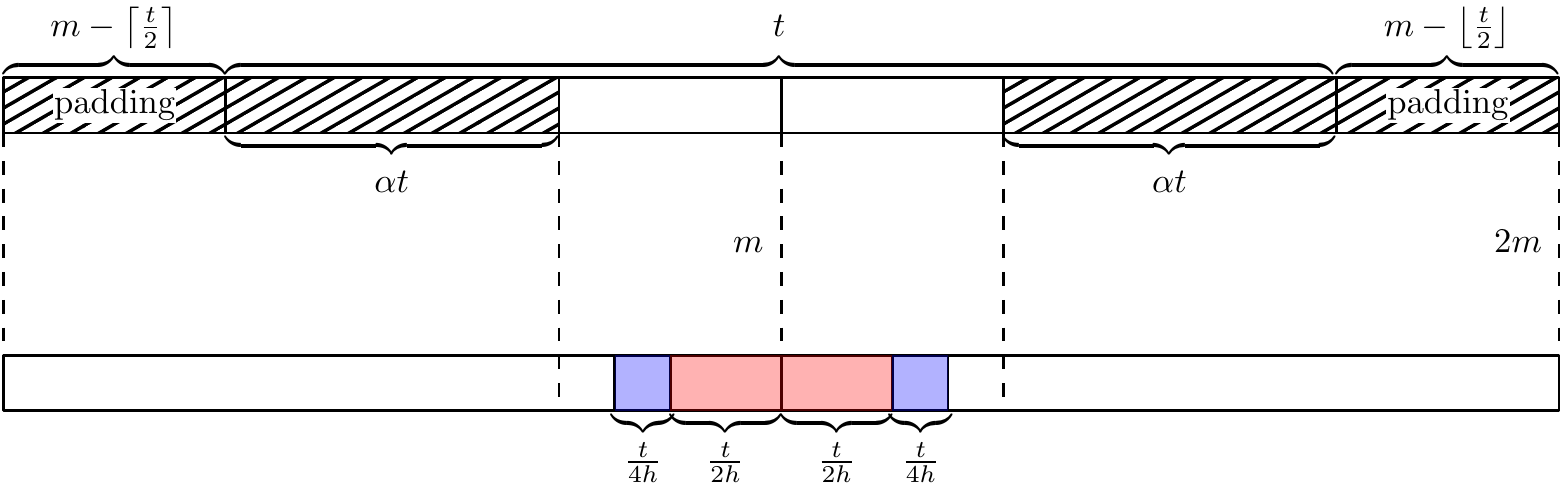}
\caption{Above: Illustration of the original union of the two input sets
  with padding elements. The players need to find elements from the unshaded region
  in the middle.
  Below: The median $x$ of $Q_A\cup Q_B$ lies within the red region.
  If the other player has an element $y$ such that $\p_\ell(y)=\p_\ell(x)$,
  then $y$ lies in the union of the red and blue regions,
  therefore it is also a valid output.}
\label{fig:2}
\end{figure}

Observe that $h = \lceil \frac{2q}{q-2p} \rceil \leq 2q$ which yields
$ 2 h^2 \leq 8q^2 \leq cn \leq t$ (recall that $n \geq 8q^2/c$). 
This implies
\begin{equation} \label{eq:rank-x_A}
  \left| \rank_{A \cup B}(x_A) -m \right| \leq h \leq \frac{t}{2h}.
\end{equation}
Recall that if  $\p_\ell(x_A) = \p_\ell(x_B)$, Bob also returns $x_B \in Q_B$
and Inequality~\eqref{eq:rank-xy} applies. 

From~\eqref{eq:rank-xy} and~\eqref{eq:rank-x_A} we deduce that
the rank of any output element $z$ satisfies (recall that $t=s+m$):
\begin{equation} \label{eq:rank}
\left| \rank_{A \cup B}(z) - m \right|
\leq \frac{t}{4h} + \frac{t}{2h} \leq \frac{t}{h}
\leq \frac{(q-2p)t}{2q}= \left(\frac12 -\alpha \right) t.
\end{equation}

As such, each output element $z$ is an $(\alpha t, \alpha t)$-mediocre element of
the original union of the two sets. 
The elements returned are $x_A$ or $x_B$ (or both). 
Alice may return $x_A$ and Bob may return $x_B$ to the processes
that have invoked their service; the elements returned by the players could be different.
Since $q=O(1)$, we have $h, \ell = O(1)$.
The number of bits exchanged is $\ell+O(1)=O(1)$ in each of the $O(\log{h})$ rounds of the protocol.
The overall communication complexity is $O(\ell \, \log{h}) =O(1)$, as claimed. 
\qed

\section{Conclusion} \label{sec:conclusion}

An obvious question is whether the three-party communication complexity of median computation
can be reduced to $O(\log{n})$.
A more general question is whether the $k$-party communication complexity
of median computation, $k \geq 3$, can be reduced to $O(k \log{n})$.
We believe that the answers to both questions are in the negative. 
Another interesting question regarding the two-party communication complexity of approximate
selection is whether the conditions in Theorems~\ref{thm:approx-k} and~\ref{thm:approx-2}
can be relaxed.

\end{document}